\documentclass[a4paper]{article}
\usepackage[a4paper]{geometry}
\usepackage{amsfonts,amssymb,amsmath,amsthm,bm}
\newcommand{\Ei}{\operatorname{Ei}}

\newtheorem{theorem}{Theorem}
\newtheorem{remark}{Remark}
\newtheorem{lemma}{Lemma}

\newtheorem{proposition}{Proposition}
\newtheorem{definition}{Definition}
\newtheorem{corollary}{Corollary}
\newcommand{\filt}{\mathcal{F}}
\newcommand{\prob}{\mathbb{P}}
\begin{document}

\title{Geometric Asian Option Pricing in General Affine
Stochastic Volatility Models with Jumps}
\author{Friedrich Hubalek, \\
{\small Vienna University of Technology, Financial and Actuarial Mathematics,
} \\
{\small Wiedner Hauptstra\ss {}e~8/105--1, A--1040 Vienna, Austria
(fhubalek@fam.tuwien.ac.at)} \\
Martin Keller-Ressel, \\
{\small Institut f\"ur Mathematik, TU Berlin,} \\
{\small Strasse des 17. Juni 136, D--10623 Berlin, Germany
(mkeller@math.tu-berlin.de)} \and Carlo Sgarra\thanks{%
Corresponding author. Tel.: +39 02 2399 4570; fax: +39 02 2399 4621} \\
{\small Department of Mathematics, Politecnico di Milano,}\\
{\small Piazza Leonardo da Vinci, 32, I--20133 Milan, Italy
(carlo.sgarra@polimi.it)}}
\maketitle

\begin{abstract}
In this paper we present some results on Geometric Asian option valuation
for affine stochastic volatility models with jumps. We shall provide a general
framework into which several different valuation problems based on some average
process can be cast, and we shall obtain close-form solutions for some relevant
affine model classes.

\end{abstract}

\section{Introduction}

Asian options are quite common derivatives often combined with
other financial claims in order to construct structured products \cite{Kat};
they can in fact provide protection against strong price fluctuations in
volatile markets and reduce the possibilities of market manipulations near
the expiry. That is because Asian options are roughly speaking options on
the average value assumed by the underlying during the option's life, and
they require some mathematical effort in order to describe the dynamics of the
average under consideration. For these reasons it is interesting to develop
realistic financial models and efficient numerical algorithms to evaluate these
kind of options.

Asian options are usually grouped into two main classes according to their
payoff: if $S_{T}$ is the value of the underlying asset at maturity $T$, $K$ is
the strike price, and $A_{T}$ is a suitably defined average of
the values assumed by the stock during the period under consideration,
the so called "Average Strike" (sometimes called "Floating Strike")
Asian Calls have the payoff given by the following expression: $\left(
S_{T}-A_{T}\right) _{+}$,\ while the payoff of the "Average Price"
(Sometimes called "Fixed Strike" or "Average Rate") Asian Calls is given by:
$\left( A_{T}-K\right) _{+}$. The average considered can be the Geometric
or the Arithmetic one and it can be calculated on a continuous or discrete
monitoring basis. All these details are specified by the contracts stipulated
by two counterparts, as Asian options are mainly OTC (Over the Counter) traded
financial derivatives.

Several results are available on Arithmetic Asian options. In the classical
Black-Scholes framework the papers by H.Geman, M. Yor \cite{Geman/Yor1993}
and by D. Dufresne \cite{Dufresne2001} present an evaluation approach based
on exponential functionals of Brownian motion properties, while in another
paper by D. Dufresne \cite{Dufresne2005}, some explicit valuation formulas
related to the Bessel process are given. More recently in the paper by M.
Schr\"oder \cite{Schroeder2008}\ a comprehensive analysis of the Arithmetic
Asian options is provided, emphasizing the role plaid by complex analysis
and special functions in solving the main valuation problems. In a more
general Exponential L\'{e}vy setting some results on Arithmetic Asian
options are included in the papers by H. Albrecher and M. Predota \cite%
{Albrecher/Predota2004}, where the L\'{e}vy process describing the
underlying evolution is assumed to be of NIG type, and by H. Albrecher \cite%
{Albrecher2004}. When explicit formulas do not exist some accurate analytic
approximations have been proposed, like in \cite{Milevsky1998}.
The paper by J. Ve\v{c}er and M. Xu \cite{Vecer/Xu2004} is the only one, to our
knowledge, dealing with the valuation problem of Arithmetic Asian options in
a general semimartingale setting, where a Partial Integro-Differential
Equation is provided solving the problem in the special case of an
underlying described by a process with independent increments. The paper by J.-P.
Fouque and C.-H. Han \cite{Fouque/Han2003} deals with the evaluation problem
for Arithmetic Asian options in a stochastic volatility framework by extending
the reduction technique introduced by J. Ve\v{c}er and M. Xu \cite{Vecer/Xu2004}.
As far as lower and upper bounds on prices are concerned some results are available
for Arithmetic Asian options both in the continuous \cite{Goovaerts2000} and
discrete monitoring case \cite{Goovaerts2006}, where a convenient use of the
comonotonicity property is exploited in order to provide such bounds. The
Hedging issue of Asian options has been considered in \cite%
{AlbrecherGoovaerts2003}, where a static strategy is examined.

As far as Geometric Asian Options are concerned, their evaluation in the
basic Black-Scholes setting is very simple. While a direct argument can
provide an explicit solution for Geometric Average Rate Calls (under continuous monitoring),
a slightly more involved calculation can provide at least an accurate
numerical approximation for the Average Strike Call (still under continuous monitoring)
in the same framework \cite{Wilmott/Dewynne/Howison}.
As far as the L\'{e}vy models are concerned, several results are also available:
we recall here the paper by C.B. Zhang and C.W. Oosterlee \cite{ZO2013}.
For the discrete monitoring case some definite results for L\'{e}vy models
are illustrated in \cite{Fusai/Meucci2008}.

It has been pointed out \cite{Glasserman} that the Geometric Asian option
pricing is extremely useful also for the arithmetic average option valuation
via Monte Carlo methods with control variables.

As far as stochastic volatility models are concerned,  while Y.L. Cheung and H.Y.
Wong \cite{Cheung/Wong2004} obtain via a perturbation method some
semi-analytical formulas for Geometric Asian options in stochastic
volatility models exhibiting a mean-reverting behavior.
The paper by I. Peng deals with Geometric Asian options valuation in a
local volatility setting, namely in the CEV model \cite{Pen2006}.
More recent results, standing on Asymptotic expansion techniques,
have been obtained by E. Gobet and M. Miri \cite{Gobet/Miri2011}.
In a very recent paper by B. Kim and I.-S. Wee \cite{KW2011} the Geometric Asian
option pricing problem has been studied for the stochastic volatility model
proposed by Heston.

The intrinsic limitations of the Black-Scholes model are well known since a
long time. In particular, the fat tail, the volatility clustering, the
aggregational gaussianity features exhibited by stock prices distributions,
and moreover the volatility smiles and the leverage effect empirically
observed cannot be explained by this model. While L\'{e}vy-based and
stochastic volatility models can explain some of these phenomena separately,
the models including both features, i.e. stochastic volatility \emph{and}
jumps, can provide a much more realistic description of stock prices
behavior. Several models of this type have been proposed in the literature
and we just mention here the models suggested by D. Bates \cite{Bates1996},
\cite{Bates2000}, by O.E. Barndorff-Nielsen and N. Shephard \cite{BNS2001},
\cite{BNNS2002}, and the Time-Changed L\'evy models proposed by P. Carr,
H. Geman, D. Madan and M. Yor  \cite{CGMY2003}, \cite{CW2004} among others.
The price to pay for this substantial improvement in modeling is a
bigger difficulty in performing calculations for evaluating derivatives.
Very few results are available in this more general setting.

In \cite{HS2011} a semi-explicit evaluation formula for
Geometric Asian Options, for fixed and floating strike, under continuous
monitoring, when both stochastic volatility and jumps come into play has
been provided; in that paper a specific model framework was considered, i.e.
the Barndorff-Nielsen and Shephard model.

Almost all of the above mentioned pricing models in which stochastic volatility
features have been combined with jumps belong to the large family of
affine models, according to the definition provided by D. Duffie, D. Filipovic
and W. Schachermayer \cite{duffie2003affine}. This class includes almost all
the most popular pricing models existing in the literature related to many
different type of underlying assets: fixed income securities, credit risk
models, equities and commodities. Many relevant features of these models
can be described in a unified way by the very general framework provided
by the affine process approach. For an extensive treatment of the general
properties of affine models and some related technical issues we mention
the Thesis by M. Keller-Ressel \cite{KR2008}.

We want to recall here another relevant class of valuation problems
requiring the description of some average process, i.e. options on realized
volatility and variance swaps. Several recent papers attacked these valuation
problems in different setting. In the paper by J. Kallsen, J. Muhle-Karbe and
M. Voss \cite{KMKV2011} the pricing of options on variance in affine stochastic
volatility models has been extensively investigated; some results in a
Barndorff-Nielsen and Shephard modeling framework were provided in \cite{BGK2007},
while in \cite{CLW2012} variance swaps pricing has been studied for time-changed
L\'{e}vy models.

The contribution of the present paper is to develop a general valuation scheme
and to provide some semi-explicit evaluation formulas for Geometric Asian Options,
when the underlying process describing the joint dynamics of logreturns and
volatility is affine. We shall provide a quite general framework into which
several different valuation problems can be formulated and solved: all those
based on the geometric mean of the variables, including average price and
average strike Asian call options; variance options valuation can be also
cast into the present framework, but only for continuous return processes,
as we shall discuss in Section \ref{integral}.

In next section we shall introduce the general setting and the notations used
throughout the paper, while in Section \ref{integral} we shall present the
introductory results on the affine representation for integral functionals.
After providing in Section \ref{changenum} an auxiliary result based on a
change-of-numeraire technique, which will turn out to be useful for Average
Strike option calculations, in Section \ref{general} we shall present the
general results on Geometric Asian options valuation in a general affine
framework.
In Section \ref{sec:Examples} we shall apply our general results to the most popular
concrete affine stochastic volatility models and we shall provide the semi-explicit
formulas for both Average Price and Average Strike options for these models.
In section \ref{concluding} we shall resume the main results obtained in this
paper and we'll outline some possible developments of the present work.

\section{Model Setup}

The purpose of this section is to clarify the framework in which
we are going to develop our pricing problem and to clarify the
basic notations adopted in the following. The results recalled
here are mainly based on the treatment provided in \cite {KR2008}
and \cite {KellerRessel2010}.
We fix some time horizon $T>0$ up to which we wish to model the price
process of some financial asset.
Let $(\Omega,\filt,(\filt_t)_{t\in[0,T]},\prob)$ be a filtered probability space,
which supports all the processes we encounter in the sequel.

We shall call an affine process
a stochastically continuous, time-homogeneous
Markov process $(X_t, \mathbb{P}^{x})$ with state space
$D= \mathbb{R}_+ ^{m} \times \mathbb{R}^{n}$ if its
characteristic function is an exponentially affine function of the state vector,
i.e., if there exist functions $\phi : \mathbb{R}_+ \times \mathcal{U} \rightarrow \mathbb{C}_-$,
${\bm\psi}:\mathbb{R}_+\times \mathcal{U} \rightarrow \mathcal{U}$ such that
\begin{equation}
\log \bigl(\mathbb{E} [\exp {\bm{u} \cdot \bm{X}_{t} }|\bm X_0] \bigr)
=\phi(t,\bm{u})+ \bm X_0\cdot {\bm\psi}(t,\bm{u})
\end{equation}
for all $(t,\bm{u}) \in \mathbb{R}_+ \times \mathcal{U}$ and where
\begin{equation}
\mathbb{C}_- := \{u \in \mathbb{C}: \Re\,u \le 0\} \quad \text{and} \quad \mathcal{U} := \mathbb{C}_-^m \times i \mathbb{R}^n.
\end{equation}
By convention, the logarithm above denotes the distinguished logarithm in complex plane, that makes $\phi$ and $\psi$ jointly continuous in the complex plane (cf. \cite{duffie2003affine}). Note that due to the Markov property an analogous equation also holds true for expectations conditional on $\bm{X}_s$, that is
\begin{equation}
\log \bigl(\mathbb{E} [\exp {\bm{u} \cdot \bm{X}_{t} }|\bm X_s] \bigr)
=\phi(t-s,\bm{u})+ \bm X_s\cdot {\bm\psi}(t-s,\bm{u})
\end{equation}
for all $0 \le s \le t$ and $\bm u \in \mathcal{U}$.
An affine process is called \textbf{regular} if the derivatives:%
\begin{equation*}
F(\bm{u}):=\frac{\partial \phi}{\partial t}(t,\bm{u})|_{t=0^{+}},\
{\mathbf R}(\bm{u}):=\frac{%
\partial {\bm\psi} }{\partial t}(t,\bm{u})|_{t=0^{+}},
\end{equation*}%
exist for all $\bm{u}\in \mathcal{U}$, and are continuous at $\bm{u}=0$. It has been shown in \cite{KellerSchachTeich2010} that \emph{any} affine process in the sense of the above definition is regular and hence that the functions $F(\bm{u})$ and $R(\bm{u})$ are well-defined.

Since the functions $F(u)$ and ${\bm R}({\bm u})$ completely characterize the process
$(\bm{X}_{t})_{t\geq 0}$ they are called the \textbf{functional characteristics of}
$(\bm{X}_{t})_{t\geq 0}$.

In the following we shall need the notion of \textbf{truncation function},
but we'll specify which truncation function will be used whenever it will be
necessary.
When an affine process will be assumed to describe the price dynamics
of some asset, we shall refer to it as an affine pricing model.

In the following we shall assume the (risk-neutral) stock price process
$S_t$ to be given as

\begin{equation}  \label{eq:AssetPrice}
S_t=\exp\bigl\{(r-q)t+X_t\},
\end{equation}

where $r$ is the risk-free interest rate, $q$ is the dividend yield and $X_t$ is the
discounted dividend-corrected log-price process.\newline
Let $V_t$ denote another (one-dimensional) process with $V_0>0$,
such that $(X_t,V_t)$ is a stochastically continuous, time-homogeneous Markov process.

We define the process $(X_t,V_t)$ an Affine Stochastic Volatility
model if the cumulant generating function of $(X_t,V_t)$ is of the special affine form
\begin{equation}  \label{eq:AffineCF}
\log\bigl(\mathbb{E}[\exp\{uX_t+wV_t\}|X_0,V_0]\bigr)%
=\phi(t,u,w)+V_0\psi(t,u,w)+X_0u.
\end{equation}
Note that this setup is as in \cite[Section 5]{KellerRessel2010}, from where
we will adopt the nomenclature and call $(X_t,V_t)$ affine stochastic volatility
(ASV) process and the associated asset price model ASV model.\newline

\begin{remark}\label{R12}
A bivariate affine model has functional characteristics $F$ and $\bm R=(R_1,R_2)$.
An ASV has $R_1=0$ and we set $R=R_2$ and call simply $F,R$ the functional characteristics.
\end{remark}

The following theorem characterizes regular ASV processes and provides a
representation result for the functions $F$, $R$.

\begin{theorem}
{\cite[Theorem 2.7]{duffie2003affine}}\label{Thm:AffineCF}
Let $\left( X_{t}, V_{t} \right) _{t\geq 0}$ be a regular ASV process.
Then there exist a set of parameters $(a,\alpha ,b,\beta ,c,\gamma ,m,\mu )$
whrer $a,\alpha $ are positive semi-definite matrices, $b,\beta \in \mathbb{R}^{2}$,
$c,\gamma \geq 0$ and $m,\mu $ are L\'{e}vy measures on $\mathbb{R}^{2}$, such that
 \begin{align*}
F(u,w)& =\frac{1}{2}(u,w)\cdot a\cdot (u,w)^{\text{T}}+b\cdot (u,w)^{\text{T}%
}-c+\int_{D\backslash \left\{ 0\right\} }\left( \text{e}%
^{xu+yw}-1-h_{F}(x,y)\cdot (u,w)^{\text{T}}\right) m(dx,dy) \\
R(u,w)& =\frac{1}{2}(u,w)\cdot \alpha \cdot (u,w)^{\text{T}}+\beta \cdot
(u,w)^{\text{T}}-\gamma +\int_{D\backslash \left\{ 0\right\} }\left( \text{e}%
^{xu+yw}-1-h_{R}(x,y)\cdot (u,w)^{\text{T}}\right) \mu (dx,dy),
\end{align*}%
holds, where $h_{F}(x,y)$, $h_{R}(x,y)$ are suitable truncation functions.
Furthermore the functions $\phi$ and $\psi$ in \eqref{eq:AffineCF} fulfill the generalized Riccati equations:
\begin{align}
\partial _{t}\phi(t,u,w)& =F(u,\psi(t,u,w)), & \phi(0,u,w)& =0,  \label{Riccatieq} \\
\partial _{t}\psi(t,u,w)& =R(u,\psi(t,u,w)), & \psi(0,u,w)& =w.  \notag
\end{align}%
\end{theorem}

%It has been proved by M.~Keller-Ressel, W.~Schachermayer and J.~Teichmann
%\cite{KellerSchachTeich2010} that the stochastic continuity and the affine
%behavior of the cumulant function are sufficient to guarantee the regularity
%property, in such a way that this condition will be automatically satisfied
%in our setting.

For option pricing we employ a structure preserving martingale measure.
This means, we choose an equivalent martingale measure, such that the model
structure remains unchanged, only model parameters change. For several
particular models enjoying the affine structure, a systematic investigation
has been performed on the class of equivalent martingale measure, also providing
a full characterization of the subclass of structure preserving measures:
for the BNS model we mention the paper by E. Nicolato and E. Venardos
\cite{NicolatoVenardos2003}, while for the Bates model a brief discussion
on the subject is included in \cite{Bates1996}.

The following proposition provides a sufficient condition for an affine
process to be conservative (i.e. non-exploding) and a martingale:

\begin{proposition}\label{martprop}
{\cite[Corollary 2.1]{KellerRessel2010}} Let $(X_t,V_t)$ be defined as
before and the quantity $\chi(u)$ be defined as follows:
\begin{equation}
\chi(u):=\frac{\partial R}{\partial w}(u,w)|_{w=0}.
\end{equation}
If $F(0,0)=R(0,0)=F(1,0)=R(1,0)=0$ and $\max\{\chi(0),\chi(1)\}<
\infty$, then $\exp\{X_t\}$ is a conservative process and a martingale.
\end{proposition}
Note that $F(0,0) = R(0,0) = 0$ is equivalent to $c = \gamma = 0$.

\section{Integral functionals for ASV models \label{integral}}
Our starting point is an affine ASV model $(X,V)$ as introduced above. To study Geometric Asian options or realized variance options we introduce the associated integral processes $Y$ and $Z$ with
\begin{equation}
Y_t=\int_0^tX_sds,\qquad
Z_t=\int_0^tV_sds.
\end{equation}

\begin{proposition}\label{main}
If $(X,V)$ is an ASV model with functional characteristics $(F,R)$,
then the joint law of $(X_t,V_t,Y_t,Z_t)$ is described by
\begin{equation}
\log E[e^{u_1X_t+u_2V_t+u_3Y_t+u_4Z_t}|X_0,V_0]=
\Phi(t,u_1,u_2,u_3,u_4)+(u_1+u_3t)X_0+\Psi(t,u_1,u_2,u_3,u_4)V_0
\end{equation}
where
\begin{align}
&\dot\Phi=F(u_1+u_3t,\Psi) && \Phi(0)=0\label{joint-law-riccati1}\\
&\dot\Psi=R(u_1+u_3t,\Psi)+u_4\label{joint-law-riccati2}&&
\Psi(0)=u_2.
\end{align}
\end{proposition}
\begin{proof}
It follows from \cite[Theorem~4.10, p.50]{KR2008} for two dimensions,
that $(X,V,Y,Z)$ is affine,
\begin{equation}
\log E[e^{u_1X_t+u_2V_t+u_3Y_t+u_4Z_t}|X_0,V_0,Y_0,Z_0]=
\Phi(t)+\psi_1(t)X_0+\psi_2(t)V_0+\psi_3(t)Y_0+\psi_4(t)Z_0,
\end{equation}
where the $\Phi$ and $\psi_i$ satisfy the Riccati equations
\begin{align}\label{Ric}
&\dot\Phi=F(\psi_1,\psi_2) && \Phi(0)=0\\
&\dot\psi_1=\psi_3 && \psi_1(0)=u_1\\
&\dot\psi_2=R(\psi_1,\psi_2)+\psi_4&& \psi_2(0)=u_2\\
&\dot\psi_3=0&& \psi_3(0)=u_3\\
&\dot\psi_4=0&& \psi_4(0)=u_4.
\end{align}
Remember that the solutions of (\ref{Ric}) depend on the parameters $u_1,u_2,u_3,u_4$,
thus $\psi_1(t)=\psi_1(t;u_1,u_2,u_3,u_4)$ etc.
Some of those equations can be immediately integrated.
Obviously $\psi_3(t)=u_3$, $\psi_4(t)=u_4$,
$\psi_1(t)=u_1+u_3t$ and
the only relevant
equations are
\begin{align}
&\dot\Phi=F(u_1+u_3t,\psi_2) && \Phi(0)=0\label{relevant1}\\
&\dot\psi_2=R(u_1+u_3t,\psi_2)+u_4&&
\psi_2(0)=u_2\label{relevant2}
\end{align}
Then we note that $Y_0=0$ and $Z_0=0$ and we set $\Psi=\psi_2$.
(\ref{joint-law-riccati1})and (\ref{joint-law-riccati2})
follow from (\ref{relevant1}) and (\ref{relevant2}).
\end{proof}

%%%%%%%%%%%%%%%%%%%%%%%%%%%%%%%%%%%%%%%%%%%%%%%%%%%%%%%%%%%%%%%%5

\begin{remark}
Variance swaps and options on realized variance in stochastic volatility models
with jumps have been studied in \cite{BGK2007} and \cite{Sep2008}.
In a general affince setting they have been investigated extensively in
the paper \cite{KMKV2011} where
the realized variance is approximated by the quadratic variation
of the log-return process. For continuous return processes, such as the Heston model, for example, the quadratic
variation $[X,X]$ and its predictable part $\langle X,X\rangle$ coincide
with integrated variance, which is our $Z$.
\end{remark}
The cumulant of the integrated variance can be computed according to
the following corollary, which turns out to be a special case
of \cite[Lemma~5.1, P.634]{KMKV2011}.
\begin{corollary}
\begin{equation}
\log E[e^{wZ_t}]=\phi(t,w)+V_0\psi(t,w)
\end{equation}
where
\begin{align}
&\dot\phi=F(0,\psi) && \phi(0)=0\\
&\dot\psi=R(0,\psi)+w&& \psi(0)=0
\end{align}
\end{corollary}
\begin{proof}
This follows immediately from Prop.\ref{main} with $u_1=0$, $u_2=0$, $u_3=0$, $u_4=w$.
\end{proof}

\section{Change of numeraire for ASV models \label{changenum}}
To calculate the price of the average strike option we apply the
change-of-numeraire technique and take the stock as a new numeraire.

From now on we denote the martingale measures with the bond resp.\ stock
as a numeraire by $Q^0$ resp.\ $Q^1$, and expectations $E^0$ resp.\ $E^1$.
\begin{equation}
\log E^0_{x,v}[e^{u_1X(t)+u_2V(t)}]=
\phi^0(t,u_1,u_2)+x\psi_1^0(t,u_1,u_2)+v\psi_2^0(t,u_1,u_2)
\end{equation}
Thus we have the density process
\begin{equation}\label{dQ1}
\frac{dQ^1}{dQ^0}(t)=e^{X_t-x}
\end{equation}
on $\mathcal F_t$.

Let's start with the following
\begin{lemma}\label{num-change}
If $(X,V)$ is affine under $Q^0$ with functional characteristics $F^0$ and $R^0$, then
it is affine under $Q^1$ with functional characteristics $F^1$ and $R^1$ given by
\begin{equation}
F^1(u_1,u_2)=F^0(u_1+1,u_2),\quad
R^1(u_1,u_2)=R^0(u_1+1,u_2)
\end{equation}
\end{lemma}
\begin{proof}
\begin{multline}
\log E^1_{x,v}[e^{u_1X(t)+u_2V(t)}]=
\log E^0_{x,v}[e^{x+X_t}\cdot e^{u_1X(t)+u_2V(t)}]=\\
-x+\log E^0_{x,v}[e^{(u_1+1)X(t)+u_2V(t)}]=\\
\phi^0(t,u_1+1,u_2)+x(\psi^0_1(t,u_1+1,u_2)-1)+v\psi^0_2(t,u_1+1,u_2)=\\
\phi^1(t,u_1,u_2)+x\psi_1^1(t,u_1,u_2)+v\psi_2^1(t,u_1,u_2)
\end{multline}
with
\begin{align}
&
\phi^1(t,u_1,u_2)
=\phi^0(t,u_1+1,u_2),\\
&
\psi_1^1(t,u_1,u_2)=
\psi_1^0(t,u_1+1,u_2)-1,\\
&
\psi_2^1(t,u_1,u_2)=
\psi_2^0(t,u_1+1,u_2)
\end{align}
Thus
\begin{equation}
F^1(u_1,u_2)=F^0(u_1+1,u_2)
\quad
R^1(u_1,u_2)=R^0(u_1+1,u_2),
\end{equation}
\end{proof}
If $e^X$ is a martingale we have $F^0(1,0) = R^0(1,0)=0$ and thus $F^1(0,0) = R^1(0,0)=0$.
\begin{lemma}\label{joint1}
If $(X,V)$ is an ASV model, then the joint law of $(X_t,Y_t)$ under $Q^1$
is described by
\begin{equation}
\log E^1[e^{uX_t+wY_t}]=\phi^1(t,u,w)+v \psi^1(t,u,w)
\end{equation}
where
\begin{align}\label{joint-law}
&(\phi^1)'=F(u+1,\psi^1) && \phi^1(0)=0\\
&(\psi^1)'=R(u+1,\psi^1)&& \psi^1(0)=w.
\end{align}
\end{lemma}
\begin{proof}
This follows from Lemma~\ref{num-change} and Proposition~\ref{main} applied to $Q^1$
resp.\ $F^1,R^1$.
\end{proof}

%%%%%%%%%%%%%%%%%%%%%%%%%%%%%%%%
\section{General results for Geometric Asian options\label{general}}
\subsection{Average price}
Let us denote by $\bar X_T$ the arithmetic average of the log-returns process
and by $\hat S_T$ the geometric average  of the stock prices, then
\begin{equation}
\bar X_T=(r - q) + \frac1T\int_0^TX_sds,\quad
\hat S_T=e^{\bar X_T} = \exp \left((r-q) + \frac1T\int_0^TX_sds \right) .
\end{equation}
For average strike we shall need the cumulant of integrated log-returns.
\begin{corollary}
If $(X,V)$ is an ASV model, then the law of $Y_t = \int_0^t X_s ds$ is described by
\begin{equation}
\log E[e^{wY_t}]=\Phi(t,w)+wtX_0+V_0\psi(t,w)
\end{equation}
where
\begin{align}
&\dot\Phi=F(wt,\psi) && \Phi(0)=0\label{joint-law-riccati1}\\
&\dot\psi=R(wt,\psi)&& \psi(0)=0\label{joint-law-riccati2}.
\end{align}
\end{corollary}
\begin{proof}
This follows immediately from Prop.\ref{main} with $u_1=0$, $u_2=0$, $u_3=w$, $u_4=0$.
\end{proof}
%%%%%%%%%%%%%%%%%%%%%%%%%%%%%%%%%%%%%%%%%%%%%%%%%%%%%%%%%%%%%%
\begin{theorem}\label{price}
Assume there exists $a>1$ such that
\begin{equation}\label{integrability-a}
E[e^{a\bar X_T}]<\infty,
\end{equation}
then the time-zero value of an average price Asian call option is given by
\begin{equation}\label{E-price}
E[e^{-rT}(\hat S_T-K)_+]
=\frac{e^{-rT}}{2\pi i}\int\limits_{a-i\infty}^{a+i\infty}
\left(\frac{1}{K}\right)^u\frac{K}{u(u-1)}
e^{\kappa(T,u)}du,
\end{equation}
with the cumulant function $\kappa(T,u)=\log E[e^{u\bar X_T}]$.
It is given by
\begin{equation}
\kappa(T,u)= u (r - q) + \phi(T,u)+uX_0+\psi(T,u)V_0,
\end{equation}
where
\begin{align}\label{Ric-Price}
&\dot\phi=F\left(\frac{ut}{T},\psi\right) && \phi(0)=0\\
&\dot\psi=R\left(\frac{ut}{T},\psi\right)&& \psi(0)=0.
\end{align}
\end{theorem}
\begin{proof}
In order to evaluate the expectation~(\ref{E-price}) we first use the integral
representation (\ref{lap-call}), which yields
\begin{equation}
(\hat S_T-K)_+
=\frac1{2\pi i}\int\limits_{a-i\infty}^{a+i\infty}
\left(\frac{1}{K}\right)^u\frac{K}{u(u-1)}
e^{u\bar X_t}du,
\end{equation}
and then apply Fubini's Theorem, see also \cite{HKK2006} and \cite{HS2011}.
\end{proof}
\begin{remark}\label{rem-inta}
The integrability condition (\ref{integrability-a}) guarantees the existence
of the cumulant function $\kappa(T,u)$ at $\Re u=a$. It will imply some restrictions on the parameters of the concrete models studied in Section~\ref{sec:Examples}. By the results in \cite{KMayerhofer2013} it is equivalent to the existence to solutions of the Riccati equations \eqref{Ric-Price} for the parameter value $u = a$.
The proper set of parameters can be determined individually for each concrete model by studying the real singularities of the cumulant functions, see \cite[Satz~3.4.1, P.153f]{Doe1}, though
we are not going to give all details for all models in the example section below.
\end{remark}
%%%%%%%%%%%%%%%%%%%%%%%%%%%%%%%%%%%%%%%%%%%%%%%%%%%%%%%%%%%%%%

\subsection{Average strike}
%%%%%%%%%%%%%%%%%%%%%%%%%%%%%%%%
\begin{theorem}
If there exists $b<0$ such that
\begin{equation}\label{integrability-b}
E[e^{b\bar X_T}]<\infty,
\end{equation}
then the time-zero value of an average strike Asian call option is given by
\begin{equation}
E[e^{-rT}(S_T-\hat S_T)_+]
=\frac{e^{-qT}}{2\pi i}\int\limits_{b-i\infty}^{b+i\infty}
\frac{1}{u(u-1)}
e^{\kappa(T,u)}du,
\end{equation}
where $\kappa(T,u)=\log E[e^{u\bar X_T+(1-u)X_T}]$.
It is given by
\begin{eqnarray}\label{kappa1}
&&\kappa(T,u)=u(r-q) + \phi(T,u)+V_0\psi(T,u)+X_0
\end{eqnarray}
where
\begin{align}\label{Ric-Strike}
&\dot\phi=F\left(\frac{ut}{T} + (1-u),\psi\right) && \phi(0)=0\\
&\dot\psi=R\left(\frac{ut}{T} + (1-u),\psi\right) && \psi(0)=0.
\end{align}
\end{theorem}
\begin{proof}
Using the change-of-numeraire technique with the density process~(\ref{dQ1})
we obtain
\begin{equation}
E^0[(S_T-\hat S_T)_+]=e^{(r-q)T}E^1[(1-e^{\bar X_T-X_T})_+].
\end{equation}
This is just the payoff of a put option on $e^{\bar X_T-X_T}$ with asset and strike both equal to $1$.

The function $\kappa$ is the cumulant function of $\bar X_T-X_T$, which can be obtained
from the joint cumulant of $Y_T$ and $X_T$ in terms of the functions $\phi$ and $\psi$ from Lemma~\ref{joint1}.

Similar to the proof of Theorem~\ref{price}, we can now apply the Laplace integral formula (\ref{lap-put}) provided in the appendix and Fubini's Theorem to obtain the result.
\end{proof}

\begin{remark}
For the integrability condition~(\ref{integrability-b})
a remark similar to Remark~\ref{rem-inta} above applies.
\end{remark}

\begin{proposition}
Average strike and price Riccati equations have the same structure
provided the parameters are changed in the following way.
$u\mapsto u + \frac{t}{T}(1-u)$.
\end{proposition}

\begin{remark}
The property just described in the proposition above is actually a particular
case of a general result called the duality principles in option pricing. This
basic property has been systematically investigated in a general semimartingale
setting in \cite{Pap2007} and in \cite{EPS2008}.
\end{remark}

%%%%%%%%%%%%%%%%%%%%%%%%%%%%%%%%

\section{Geometric Asian options for concrete affine stochastic volatility
models}\label{sec:Examples}
We now discuss some popular ASV models from the finance literature
(for a very nice summary and many more examples, the interested reader
is refered to \cite{Kallsen2006ASV}).
For a few relevant cases we will obtain an explicit solution of the
corresponding Riccati equations. Let us recall, that for all models
the asset price will be modeled by $S_t=e^{(r - q)t + X_t}$, where $X$ denotes the discounted
log-price.

In the following examples we shall continue to assume that the model parameters
will verify the conditions in Proposition~\ref{martprop}, and consequently $e^{X_t}$ is a martingale.

\subsection{Heston model}
The Heston \cite{Heston} model describes the volatility dynamics by means of
a CIR-type stochastic differential equation with mean reversion.

The evolution of the discounted log-returns under the risk-neutral measure is then
given by
\begin{equation}\label{Dynamics1}
dX_t=\left( -{\frac12}V_t\right)dt+\sqrt{V_t}dW_t^1,
\end{equation}
\begin{equation}
dV_t=\lambda(\theta-V_t)dt+\zeta\sqrt{V_t}dW_t^2,
\label{Dynamics2}
\end{equation}%
where $\lambda$, $\theta$, and $\zeta$ are strictly positive parameters.
Moreover, in (\ref{Dynamics1}) and (\ref{Dynamics2}) $W^1$ and $W^2$ are
standard Wiener processes having constant correlation $\rho\in[-1,+1]$.

It can be shown that, if the following condition is satisfied:
\begin{equation}
\zeta^2<2\lambda\theta, \label{Parameters' restrict}
\end{equation}%
then the volatility process $V$ remains strictly positive (see \cite{Fel1951}).

The affine characteristics are \cite{KR2008,KellerRessel2010}
\begin{equation}\label{FC-Heston}
F(u,w)=\lambda\theta w,\quad
R(u,w)=\frac12(u^2-u)+\frac{\zeta^2}{2}w^2-\lambda w+uw\rho\zeta.
\end{equation}
\subsubsection*{Average price}
Combining (\ref{FC-Heston}) and (\ref{Ric-Price}) we obtain the Riccati
equation for the average price
\begin{align}
&\dot\phi=\lambda\theta\psi,\quad \phi(0)=0\\
&\dot\psi=\frac{\zeta^2}{2}\psi^2-(\lambda-\rho\zeta ut/T)\psi+\frac12ut/T(ut/T-1),\quad
\psi(0)=0
\end{align}
By using a standard substitution
\begin{equation}\label{hyper-trafo}
\psi(t)=\frac2{\zeta^2}\frac{y'(t)}{y(t)}
\end{equation}
the Riccati equation can be transformed into a
linear differential equation of second order \cite{Reid1972,PZ2003}
\begin{equation}
y''+(\lambda-\rho\zeta ut/T)y'+\frac{\zeta^2}{4}ut/T(ut/T-1)=0.
\end{equation}
The general solution of this equation can be written as a linear combination
of two confluent hypergeometric functions of the first kind \cite{Sla1960}
(most commonly denoted by $_{1}F_{1} (a,b,c)$ )
\begin{equation}
y(t)=C_1y_1(t)+C_2y_2(t)
\end{equation}
\begin{equation}
y_1(t)=AM\left(a_1,\frac12,c\right)
\end{equation}
\begin{equation}
y_2(t)=ABM\left(a_1+\frac12,\frac32,c\right)
\end{equation}
where $a1$, $c$ and $A$,$B$ are defined by the following exressions:
\begin{equation}
a_1=\frac18\frac{-\zeta^2-2\rho \zeta(\xi u/T-\lambda)-2\lambda^2}{\zeta u/T \xi^{\frac32}}+\frac14
\end{equation}
\begin{equation}
c=\frac12\frac{((1+tu/T \zeta)\xi-\lambda \rho)^2}{\zeta u/T \xi^{\frac32}}
\end{equation}
\begin{equation}
\xi=\rho^2-2
\end{equation}
\begin{equation}
A=\exp{-\frac14 [t(2\lambda-\rho \zeta ut/T+\frac{\rho^2 \zeta ut/T-2\lambda \rho+2\zeta (1-ut/T)}{\xi^{\frac12}}]}
\end{equation}
\begin{equation}
B=\xi ut/T \zeta-\xi +\lambda \rho
\end{equation}
By taking into account the initial condition $\psi(0)=0$ we obtain
\begin{equation}\label{Heston-Solution1}
\psi(t)=-\frac2{\zeta^2}
\frac{y_2'(0)y_1'(t)-y_1'(0)y_2'(t)}{y_2'(0)y_1(t)-y_1'(0)y_2(t)}.
\end{equation}
In view of the (\ref{hyper-trafo}) we express also $\phi$ explicitly by
hypergeometric functions, namely
\begin{equation}\label{Heston-Solution2}
\phi(t)=-\lambda\theta
\frac2{\zeta^2}\ln\frac{y_2'(0)y_1(t)-y_1'(0)y_2(t)}{y_2'(0)y_1(0)-y_1'(0)y_2(0)}
\end{equation}
\subsubsection*{Average strike}
Combining (\ref{FC-Heston}) and (\ref{Ric-Strike}) with $(u,w)\mapsto(-u,u/T)$
we obtain
\begin{align}
&\dot\phi=\lambda\theta\psi+q-r && \phi(0)=0\\
&\dot\psi=\frac12u^2(t/T-1)^2+u(t/T-1)
+\frac{\zeta^2}{2}\psi^2-\lambda\psi
+\rho\zeta(u(t/T-1)+1)\psi.
&& \psi(0)=0
\end{align}
The solution to these equations can be obtained in a similar way, providing
the expressions for $\phi$ and $\psi$
analogous to (\ref{Heston-Solution1}) and (\ref{Heston-Solution2})
with $y_1,y_2$ replaced by
\begin{equation}\label{Heston-barphi1}
\bar y_1(t)=M(\bar a_1,\frac12,\bar c)
\end{equation}
and
\begin{equation}\label{Heston-barphi2}
\bar y_2(t)=\bar A\bar B M(\bar a_1+\frac12,{\textstyle\frac32},\bar c).
\end{equation}
where $\bar a_1$, $\bar c$, $\bar A$ and $\bar B$ are now defined by:
\begin{equation}
\bar a_1=\frac18\frac{\rho \zeta(2 \xi u+\lambda T)-(\lambda^2+\frac14 \zeta^2)T}{\zeta u \xi^{\frac32}}+\frac14
\end{equation}
\begin{equation}
\bar c=\frac12\frac{[(\rho^2 (u-1)+(\frac12-u)]T+\zeta ut(1-\rho^2 )+\lambda \rho T)^2}{\zeta u T \xi^{\frac32}}
\end{equation}
\begin{equation}
\bar A=\exp{\frac12 [\frac{t}{T}(\rho \zeta ((u-1)T-\frac12 ut)+\lambda T) ut
+\frac{((\rho^2 (u-1)-u+\frac12) -\frac12 \zeta u\frac{t}{T} (\rho^2-1)) + \rho \lambda )}{\xi^{\frac12}}]}
\end{equation}
\begin{equation}
\bar B=\{[(u-1) \rho^2-u+\frac12]T+ut(1- \rho^2)\}\zeta + \lambda \rho T
\end{equation}

\begin{remark}
The pricing of Geometric Asian options in Heston's model has been investigated by Kim and Wee in \cite{KW2011}. They
express the joint moment generating function of returns and integral average in terms of some series expansions. In fact, their series can be summed in closed form in terms of hypergeometric functions and agrees with our results above.
\end{remark}

\subsection{The Bates model}
In a model proposed by Bates \cite{Bates1996}, a jump component is
introduced in the previous dynamics for the log-returns by means of the
compound Poisson process $Z$:
\begin{equation}
Z_t=\sum_{i=1}^{N_t}J_i,
\end{equation}%
where $N$ is a standard Poisson process with intensity~$\nu >0$ and $%
(J_{i})$, $i=1,2,3,\ldots $, are independent random variables, all having a
normal distribution with mean $\gamma$ and standard deviation $\delta$.
In such a case the L\'{e}vy measure of $Z$ is given by:%
\begin{equation}
U(dx)={\frac{\nu }{\delta \sqrt{2\pi }}\exp \left[ -\frac{\left(
x-\gamma \right) ^{2}}{2\delta ^{2}}\right] },  \label{Normal}
\end{equation}
and the cumulant function of $Z$ takes the form:%
\begin{equation}
\kappa (z)=\nu(e^{\gamma z+\delta^2z^2/2}-1).
\end{equation}
The dynamics of discounted log-returns under the risk-neutral measure is then given by:
\begin{equation}
dX_t=(-\kappa(1)-\frac12V_t)dt+\sqrt{V_t}dW_t^1+dZ_t,  \label{Dynamics3}
\end{equation}%
and the dynamics of the volatility is the same as that proposed by the
Heston model, namely
\begin{equation}
dV_t=\lambda(\theta-V_t)dt+\zeta\sqrt{V_t}dW_t^2.
\end{equation}
The affine characteristics are
\begin{equation}\label{FC-Bates}
F(u,w)=\lambda\theta w+\kappa(u) - u\kappa(1),\quad
R(u,w)=\frac12(u^2-u)+\frac{\zeta^2}{2}u_2^2-\lambda w+uw\rho\zeta.
\end{equation}
If $\nu\to0$, then we obtain the Heston stochastic volatility model \cite{Heston}.
If $\zeta\to0$ and $V_0=\theta$ then $V_t=\theta$ we obtain the Merton jump-diffusion model
\cite{Merton1976}. Consequently we might consider the Bates model as an
extension of a Merton model to the case of stochastic volatility, or as an
extension of the Heston stochastic volatility model to the case of jumps in
the asset prices.

In the Bates model the Riccati equations for the average price are
\begin{align}
&\dot\phi=\lambda\theta\psi+\kappa(ut)-ut \kappa(1),\quad \phi(0)=0\\
&\dot\psi=\frac{\zeta^2}{2}\psi^2-(\lambda-\rho\zeta ut)\psi+\frac12ut(ut-1),\quad
\psi(0)=0.
\end{align}
We observe that the equation for $\psi$ is exactly the same as in the Heston
model and and $\phi$ equals the corresponding quantity from the Heston model
plus an integral of the cumulant of the jumps (quite easy to compute):
\begin{equation}
\phi(t)=\phi_H(t)+\int_0^t\kappa(us)ds - ut \kappa(1),
\end{equation}
where $\phi_H$ is given above in (\ref{Heston-Solution2}).
This is due to the fact, that the jumps are independent of
the continuous part.

For the average strike we obtain the same $\psi$ as in the Heston model,
while the $\phi$ is provided by the following expression (also easy to compute):
\begin{equation}
\phi(t)=\phi_H(t)+\int_0^t
\kappa\left(u\left(\frac{t}{T}-1\right)u+1\right)ds - \frac{ut}{T}\kappa(1) + (1-u)\kappa(1),
\end{equation}
where $\phi_H$ is given above, using
(\ref{Heston-barphi1})and (\ref{Heston-barphi2}).

\subsection{The Turbo-Bates model}

In \cite{Bates2000} Bates introduced a refinement of the previous model
with state-dependet jump intensity. Following \cite[Sec.6.2]{KellerRessel2010}
we will consider a simplified version with only one variance factor.
The risk-neutral dynamics for log-returns are given by
\begin{equation}
dX_t=\left(-\nu_0 \kappa(1) - \left(\frac{1}{2} + \nu_1 \kappa(1)\right) V_{t}\right)dt+\sqrt{V_t}
dW_t^1+\int_{D}x\tilde{N}(V_t,dt,dx)
\end{equation}%
\begin{equation*}
dV_t = -\lambda\left(V_t-\theta\right)dt+\zeta\sqrt{V_t}dW_t^2
\end{equation*}
where $\lambda ,\theta ,\zeta >0$ as before and the Brownian motion
$W^1,W^2$ are correlated with correlation coefficient $\rho$.
The jump component is given by $\tilde{N}(V_{t},dt,dx)=N(V_{t},dt,dx)-\mu
(V_t,dt,dx)$, where $N(V_{t},dt,dx)$ is a Poisson random measure and its
predictable compensator $\mu (V_{t},dt,dx)=(\nu_0+\nu_1V_t)F(dx)dt$,
and $F$ is some fixed jump size distribution.

The affine characteristics are
\begin{equation}\label{FC-Turbo-Bates}
F(u,w)=\nu_0\kappa(u) - u\nu_0 \kappa(1) 0+ \lambda\theta w,\quad
R(u,w)=\frac12(u^2-u)+\frac{\zeta^2}{2}w^2-\lambda w+\rho\zeta uw+\nu_1\kappa(u) - u \nu_1 \kappa(1)
\end{equation}
where $\kappa(u)$ is the cumulant generating function of $F$.

The Riccati equations for the average price are
\begin{align}
&\dot\phi=\lambda\theta\psi+\nu_0\kappa(ut/T) - ut/T \nu_0 \kappa(1),\quad \phi(0)=0\\
&\dot\psi=\frac{\zeta^2}{2}\psi^2-(\lambda-\rho\zeta ut/T)\psi
+\frac12 ut/T (ut/T -1)+\nu_1\kappa(ut/T) - ut/T \nu_1 \kappa(1),\quad
\psi(0)=0
\end{align}
and for the average strike
\begin{align}
\dot\phi &= \alpha\theta\psi+\nu_0\kappa(u(t/T-1)+1) - \left(\frac{ut}{T} + (1-u)\right) \nu_0 \kappa(1),
\quad\phi(0)=0\\
\dot\psi &= \frac12u^2(t/T-1)^2+u(t/T-1)
+\frac{\zeta^2}{2}\psi^2-\beta\psi+\rho\zeta(u(t/T-1)+1)\psi
+ \\ &+ \lambda_1\kappa((u(t/T-1)+1)) - \left(\frac{ut}{T} + (1-u)\right) \nu_1 \kappa(1),\\
\quad\psi(0)=0.
\end{align}

\subsection{Barndorff-Nielsen-Shephard model}
The BNS model has been introduced by Ole Barndorff-Nielsen and Neil Shephard.
\cite{BNS2001} , \cite{BNNS2002},
\cite{NicolatoVenardos2003}, \cite{HubalekSgarra2009}.
The model is constructed from a subordinator, called background driving L\'{e}vy process (BDLP),
with cumulant generating function
\begin{equation}
\kappa (\theta )=\log E[e^{\theta Z(1)}],
\end{equation}%
which exists for $\Re (\theta )<\ell $ with some real number $\ell >0$. The
instantaneous variance process $(V(t),t\geq 0)$ is described by the
following stochastic differential equation of Ornstein-Uhlenbeck type,
\begin{equation}
dV(t)=-\lambda V(t-)dt+dZ_{\lambda }(t),  \label{dV}
\end{equation}%
with $V_{0}>0$ and $\lambda >0$ given real numbers. The logarithmic return
process $(X(t),t\geq 0)$ is given by:
\begin{equation}
dX(t)=(-\kappa(\rho) -\frac{1}{2}V(t-))dt+\sqrt{V(t-)}dW(t)+\rho dZ_{\lambda }(t),\quad
X(0)=0,  \label{dX}
\end{equation}%
with parameters $\mu \in $, $\beta \in $, $\rho \leq 0$.
The affine characteristics are
\begin{equation}\label{FC-BNS}
F(u,w)=\lambda k(w+\rho u)-u \lambda k(\rho),\quad
R(u,w)=\frac12(u^2-u)-\lambda w.
\end{equation}
Riccati equations for average price are
\begin{align}
&\dot\phi=\lambda k(\psi+\rho ut/T)- ut/T \lambda k(\rho) && \phi(0)=0\\
&\dot\psi=\frac12(u^2t^2/T^2- ut/T)-\lambda\psi&& \psi(0)=0.
\end{align}
We remark the equation for $\psi$ is linear and can be solved explicitly,
giving
\begin{equation}
\psi(t)=
\frac{u^2}{2 T^2}f_2(t)-\frac{u}{tT}f_1(t).
\end{equation}
with
\begin{equation}
f_0(t)=\frac{1-e^{-\lambda t}}{\lambda},\qquad
f_1(t)=\frac{t}{\lambda}-\frac{1-e^{-\lambda t}}{\lambda^2},\qquad
f_2(t)=
\frac{t^2}{\lambda}-\frac{2t}{\lambda^2}
+\frac{2(1-e^{-\lambda t})}{\lambda^3}.
\end{equation}

The equation for $\phi$ yields an integral
\begin{equation}
\phi(t)=\int_0^t\lambda k\left(\psi(s)+\rho us/T \right)
-\frac{t^2}2 \frac{u}T \lambda k(\rho)
\end{equation}
Riccati equations for average strike are
\begin{align}
&\dot\phi=\lambda k(\psi+\rho(u(t/T-1)+1)) && \phi(0)=0\\
&\dot\psi=\frac12u^2(t/T-1)^2+u(t/T-1)-\lambda\psi&& \psi(0)=0.
\end{align}
The solution is quite analogous, now with
\begin{equation}
\psi(t)=\frac{u^2}{2T^2}f_2(T)+\frac{u}{T}\left(\frac12-u\right)f_1(T)
+\left(\frac{u^2}{2}-\frac{u}{2}\right)f_0(T).
\end{equation}
and the integral
\begin{equation}
\phi(t)=\int_0^T
\lambda k\left(
\psi(s)+\rho \left(\left(\frac{s}{T}-1\right)u+1\right)\right)ds
-\left(\left(1-T\right)u+T\right)\lambda k(\rho).
\end{equation}
The results agree\footnote{Actually term $-u/2$ is missing in \cite[(47)]{HS2011}
and should be included there.}
with those from \cite{HS2011}, which were obtained by a
different technique without employing the general affine framework and Riccati
equations.

\subsection{OU time-changed L\'evy processes}
Time-changed Levy processes have been introduced
by P.~Carr, H.~Geman, D.~Madan and M.~Yor \cite{CGMY2003}
in order to improve Levy models
performances in describing asset price dynamics. We shall concentrate our
attention on time changes based on processes satisfying a stochastic
differential equation of an Ornstein-Uhlenbeck or a square-roote (CIR)
type.
Let $L$ be a L\'{e}vy process with cumulant function
\begin{equation}
\theta(u)=\log E[e^{uL(1)}].
\end{equation}
Then we define
\begin{equation}
X_t=L(\Gamma(t)),  \label{eq:OUAssetPrice}
\end{equation}%
where $\Gamma(t)$ is a non-negative increasing process independent of $L$.
Here we would like to use a very popular time change,
namely an integrated Ornstein-Uhlenbeck (OU) type process.
\begin{definition}[OU time-change]
The OU time-change model is given as
\begin{equation}\label{Gamma}
\Gamma(t)=\int_0^tV(s)ds,
\end{equation}
where $V$ is now given as solution of the SDE
\begin{equation}
dV(t)=-\lambda V(t)dt+dU(t),
\end{equation}
with U being a pure jump subordinator with cumulant function $\kappa(u)$.
\end{definition}
The affine characteristics are
\begin{equation}\label{FC-OU-time-changed}
F(u,w)=\lambda\kappa(w),\quad
R(u,w)=-\lambda w+\theta(u).
\end{equation}

For the Riccati equations we get from (\ref{Ric-Price}) and (\ref{FC-OU-time-changed})
\begin{align}
&\dot\phi=\lambda\kappa(\psi) && \phi(0)=0\\
&\dot\psi=-\lambda\psi+\theta(ut)&& \psi(0)=0.
\end{align}

The Riccati equations for average strike are
\begin{align}\label{Ric-Strike}
&\dot\phi=\lambda\kappa(\psi)+q-r && \phi(0)=0\\
&\dot\psi=-\lambda\psi+\theta(u(t/T-1)+1)&& \psi(0)=0
\end{align}

Let us consider a concrete example of a time-changed L\'{e}vy given by a Kou
double exponential L\`{e}vy process with time change implied by an
integrated OU process. By recalling that the cumulant for the double
exponential has the following expression:%
\begin{equation}
\kappa (u)=\nu u\left[
\frac{p}{\alpha_{+}-u}-\frac{1-p}{\alpha_{-}+u}
\right],
\end{equation}
where $\nu$ is the intensity of the jump process,
$\alpha_{-},\alpha_{+}$ describe the exponential tails,
the Riccati equations for average price have the following explicit solution:
\begin{align}
\psi(t) =
&-e^{-\lambda t}\frac{\nu}{\lambda u}
\left\{[p\Ei(1,-\frac{\lambda \alpha_-}{u})e^{(\frac{\lambda \alpha _-}{u})}\lambda \alpha_- \Ei(1,\frac{\lambda
\alpha _+}{u})
e^{(\frac{\lambda \alpha _+}{u})}\lambda \alpha_+ +u+p \lambda
\alpha _+ \Ei(1,\frac{\lambda \alpha _+}{u})
e^{(\frac{\lambda \alpha _+}{u})}\right\} \\
&-e^{-\lambda t}\frac{\nu}{\lambda u}
\left\{[p\Ei(1,-\frac{\lambda (\alpha _-+ut)}{u})
e^{(\frac{\lambda (\alpha _- +ut)}{u})}\lambda \alpha_-
\Ei(1,\frac{\lambda (\alpha _+ -ut)}{u})
e^{(\frac{\lambda (\alpha _+ -ut)}{u})}
\lambda \alpha_+ \right.\\
&\left.+u+p \lambda \alpha _+ \Ei(1,\frac{\lambda (\alpha _+ -ut)}{u})
e^{(\frac{\lambda (\alpha _+ -ut)}{u})}\right\}
\end{align}
While the average strike Riccati equations have the following solution:
\begin{equation}
\psi(t) =
e^{-\lambda t}\frac{\nu}{u}
[p\lambda \alpha_ - T\Ei(1,-\frac{\lambda \alpha_- T-\lambda u T -\lambda u (T-t)}{u})e^{(\frac{\lambda \alpha _ - T-\lambda u T -\lambda u (T-t)}{u})}+u]
\end{equation}

\begin{equation*}
-e^{-\lambda t}\frac{\nu}{u}[p \lambda
\alpha _+ T\Ei(1,\frac{\lambda \alpha _+T-\lambda u T -\lambda u (T-t)}{u})
e^{(\frac{\lambda \alpha _+ T-\lambda u T -\lambda u (T-t)}{u})}]+
\end{equation*}

\begin{equation*}
-e^{-\lambda t}\frac{\nu}{u} [-p\lambda \alpha _+ T\Ei(1,-\frac{\lambda (\alpha _- T-\lambda u T -\lambda u (T-t))}{u})
e^{(\frac{\lambda (\alpha _- T-\lambda u T -\lambda u (T-t))}{u})}]
\end{equation*}

\subsection{CIR time-changed L\'evy processes}
Another time change which has been proposed in \cite{CGMY2003} for a Levy process in order
to improve its performances in describing logreturns statistical behavior,
is that driven by an integrated CIR process, i.e., a process satisfying
the following SDE:
\begin{equation}
dV_{t}=-\lambda \left( V_{t}-\theta \right) dt+\eta \sqrt{V_{t}}dW_t.
\end{equation}
The time-change and the returns process will be given by
(\ref{Gamma}) and (\ref{eq:OUAssetPrice}) as above.

The affine characteristics are
\begin{equation}\label{FC-CIR-time-changed}
F(u,w)=\lambda\theta w,\quad
R(u,w)=\frac{\eta^2}{2}w^2-\lambda w+\kappa(u).
\end{equation}
where $\kappa(u)$ is the cumulant generating function of the
L\'evy process.

From (\ref{Ric-Price}) and (\ref{FC-CIR-time-changed}) we obtain
the Riccati equations for average price
\begin{align}
&\dot\phi=\lambda\theta\psi && \phi(0)=0\\
&\dot\psi=\frac{\eta^2}{2}\psi^2+-\lambda\psi+\kappa(ut) && \psi(0)=0.
\end{align}

For average strike
\begin{align}
&\dot\phi=\lambda\theta\psi+q-r && \phi(0)=0\\
&\dot\psi=\frac{\eta^2}{2}\psi^2+-\lambda\psi+\kappa(u(t/T-1)+1)&& \psi(0)=0
\end{align}
Let us consider a concrete example of a time-changed L\'{e}vy given by a Kou
double exponential L\`{e}vy process with time change implied by an
integrated CIR process. By recalling that the cumulant for the double
exponential has the following expression:%
\begin{equation}
\kappa (u)=\nu u\left[
\frac{p}{\alpha_{+}-u}-\frac{1-p}{\alpha_{-}+u}
\right],
\end{equation}
where $\nu$ is the intensity of the jump process,
$\alpha_{-},\alpha_{+}$ describe the exponential tails,
the Riccati equation for $\psi$
becomes:
\begin{equation}
\dot\psi=\frac{\eta^2}{2}-\nu\dot\psi
\psi ^{2}+\lambda u\left[ \frac{p}{\nu _{-}-u}-\frac{1-p}{\nu _{+}+u}%
\right] ,\psi _{w}(0)=0.
\end{equation}%

For a symmetric jump distribution, i.e., $p=1/2$ and $\alpha_+=\alpha_-$
we can provide an explicit solution in terms of {\em Heun's confluent hypergeometric
function} $C$, see \cite{Ron1995,SK2010}:
\begin{equation}
\psi(t)=-\frac2{\eta^2}
\frac{y_2'(0)y_1'(t)-y_1'(0)y_2'(t)}{y_2'(0)y_1(t)-y_1'(0)y_2(t)}.
\end{equation}
\begin{equation}
y_1 =\exp(-\frac{\lambda t}{2}) (\alpha_+^2 -u^2 t^2) C(0,-\frac12 , 1, \frac{-\lambda^2 \alpha_+^2 +2 e t^2 \nu \alpha_+}{16u^2}, \frac{8u^2+\lambda^2 \alpha_+^2}{16u^2},\frac{u^2t^2}{\alpha_+^2})
\end{equation}
\begin{equation}
y_2 =\exp(-\frac{\lambda t}{2}) (\alpha_+^2 -u^2 t^2) C(0,+\frac12 , 1, \frac{-\lambda^2 \alpha_+^2 +2 e t^2 \nu \alpha_+}{16u^2}, \frac{8u^2+\lambda^2 \alpha_+^2}{16u^2},\frac{u^2t^2}{\alpha_+^2})
\end{equation}

\bigskip

\section{Concluding Remarks \label{concluding}}

In this paper we have just provided a framework for Geometric Asian options valuation
and we have shown that this valuation problem, through the general affine
process approach, can be reduced to solving some generalized Riccati equations
and that in many relevant cases these equations admit close-form solutions.
The final step of the present valuation procedure requires the numerical
inversion of a Laplace transform. This computation, which has become quite
standard in option pricing, nevertheless requires some care especially when
complicated special functions, like those considered insofar, are involved.
The research of a fast and accurate algorithm providing such inversion will
be the subject of our future investigation together with an extensive comparison
of the numerical methods available for Geometric Asian option pricing in affine
stochastic volatility models.
As we mentioned in Section~\ref{general}, proper integrability conditions must be verified in
order to apply our general pricing results: the existence of all the involved
cumulant functions must be assured; this can be investigated through the analysis
of the singularities of the special functions introduced.
This subject, together with a systematic numerical illustration of the present
results will be the subject of our future investigation and will be collected
in a separate paper.

\bigskip

\bigskip

\bigskip

\bigskip

\appendix
\section{Laplace formulae}
\begin{lemma}
Suppose we are given real numbers $S_0>0$, $K>0$, and $a>1$, $0<b<1$, $c<0$.
Then we have for all $x\in\mathbb R$ the formulas
\begin{equation}\label{lap-call}
(e^x-K)_+
=\frac1{2\pi i}\int\limits_{a-i\infty}^{a+i\infty}
\left(\frac{1}{K}\right)^u\frac{K}{u(u-1)}e^{u x}du,
\end{equation}
\begin{equation}\label{lap-put}
(K-e^x)_+
=\frac1{2\pi i}\int\limits_{c-i\infty}^{c+i\infty}
\left(\frac{1}{K}\right)^u\frac{K}{u(u-1)}e^{u x}du,
\end{equation}
and
\begin{equation}\label{lap-prot}
(e^x-K)_+-e^x=
\frac1{2\pi i}\int\limits_{b-i\infty}^{b+i\infty}
\left(\frac{1}{K}\right)^u\frac{K}{u(u-1)}e^{u x}du.
\end{equation}
\end{lemma}
\begin{proof}
Let $f(x)=(e^x-K)_+$. An elementary calculation provides the (bilateral)
Laplace transform of $f$, namely
\begin{equation}
\int_{-\infty}^{+\infty}f(x)e^{-u x}dx=
\left(\frac{1}{K}\right)^u\frac{K}{u(u-1)}
\end{equation}
for $\Re u>1$.
Now $f$ is continuous and has locally bounded variation, which are sufficient
conditions to guarantee that the Laplace inversion integral (with Bromwhich contour) yields the original function, that is (\ref{lap-call}). See \cite[Satz 4.4.1, P.210]{Doe1}.
The proof for (\ref{lap-put}) and (\ref{lap-prot}) is similar.
\end{proof}

\newpage

\bibliographystyle{alpha}
\bibliography{asaff}

\newcommand{\etalchar}[1]{$^{#1}$}
\begin{thebibliography}{KMKV11}

\bibitem[AG03]{AlbrecherGoovaerts2003}
Hansj{\"o}rg Albrecher and Marc Goovaerts.
\newblock Static hedging of {A}sian options under {L}\'evy models: The
  comonotonicity approach.
\newblock {Preprint}, Graz University of Technology, 2003.

\bibitem[Alb04]{Albrecher2004}
Hansj{\"o}rg Albrecher.
\newblock The valuation of {A}sian options for market models of exponential
  {L\'evy} type.
\newblock Preprint, 2004.
\newblock check!

\bibitem[AP04]{Albrecher/Predota2004}
Hansj{\"o}rg Albrecher and Martin Predota.
\newblock On {A}sian option pricing for {NIG} {L}\'evy processes.
\newblock {\em Journal of Computational and Applied Mathematics},
  172(1):153--168, 2004.

\bibitem[Bat96]{Bates1996}
D.~Bates.
\newblock Jumps and stochastic volatility: the rexchange rate processes
  implicit in deutsche mark options.
\newblock {\em The Review of Financial Studies}, 9:69--107, 1996.
\newblock check!

\bibitem[Bat00]{Bates2000}
D.~Bates.
\newblock Post-'87 crash fears in the s\&p 500 futures option market.
\newblock {\em Journal of Econometrics}, 94(1--2):181--238, 2000.
\newblock check!

\bibitem[BGK07]{BGK2007}
Fred~Espen Benth, Martin Groth, and Rodwell Kufakunesu.
\newblock Valuing volatility and variance swaps for a non-{G}aussian
  {O}rnstein-{U}hlenbeck stochastic volatility model.
\newblock {\em Applied Mathematical Finance}, 14(4):347--363, 2007.

\bibitem[BNNS02]{BNNS2002}
Ole~E. Barndorff-Nielsen, Elisa Nicolato, and Neil Shephard.
\newblock Some recent developments in stochastic volatility modelling.
\newblock {\em Quantitative Finance}, 2(1):11--23, 2002.
\newblock Special issue on volatility modelling.

\bibitem[BNS01]{BNS2001}
Ole~E. Barndorff-Nielsen and Neil Shephard.
\newblock Non-{G}aussian {O}rnstein-{U}hlenbeck-based models and some of their
  uses in financial economics.
\newblock {\em Journal of the Royal Statistical Society. Series B. Statistical
  Methodology}, 63(2):167--241, 2001.

\bibitem[CGMY03]{CGMY2003}
Peter Carr, H{\'e}lyette Geman, Dilip~B. Madan, and Marc Yor.
\newblock Stochastic volatility for {L}\'evy processes.
\newblock {\em Mathematical Finance}, 13(3):345--382, 2003.

\bibitem[CLW12]{CLW2012}
Peter Carr, Roger Lee, and Liuren Wu.
\newblock Variance swaps on time-changed {L}\'evy processes.
\newblock {\em Finance and Stochastics}, 16(2):335--355, 2012.

\bibitem[CW04a]{CW2004}
Peter Carr and Liuren Wu.
\newblock Time-changed {L}\'evy processes and option pricing.
\newblock {\em Journal of Financial Economics}, 71(1):113--141, 2004.

\bibitem[CW04b]{Cheung/Wong2004}
Ying~Lok Cheung and Hoi~Ying Wong.
\newblock Geometric {A}sian options: valuation and calibration with stochastic
  volatility.
\newblock {\em Quantitative Finance}, 4(3):301--314, 2004.

\bibitem[DFS03]{duffie2003affine}
D.~Duffie, D.~Filipovi{\'c}, and W.~Schachermayer.
\newblock Affine processes and applications in finance.
\newblock {\em The Annals of Applied Probability}, 13(3):984--1053, 2003.

\bibitem[Doe71]{Doe1}
Gustav Doetsch.
\newblock {\em Handbuch der {L}aplace-{T}ransformation. {B}and {I}: {T}heorie
  der {L}aplace-{T}ransformation}.
\newblock Birkh\"auser Verlag, Basel, 1971.

\bibitem[Duf01]{Dufresne2001}
Daniel Dufresne.
\newblock The integral of geometric {B}rownian motion.
\newblock {\em Advances in Applied Probability}, 33(1):223--241, 2001.

\bibitem[Duf05]{Dufresne2005}
Daniel Dufresne.
\newblock Bessel processes and {A}sian options.
\newblock In Mich\`ele Breton and Hatem Ben-Ameur, editors, {\em Numerical
  methods in finance}, volume~9 of {\em GERAD 25th Anniversary Series}, pages
  35--57. Springer, New York, 2005.

\bibitem[EPS08]{EPS2008}
Ernst Eberlein, Antonis Papapantoleon, and Albert~N. Shiryaev.
\newblock On the duality principle in option pricing: semimartingale setting.
\newblock {\em Finance and Stochastics}, 12(2):265--292, 2008.

\bibitem[Fel51]{Fel1951}
William Feller.
\newblock Two singular diffusion problems.
\newblock {\em Annals of Mathematics}, 54:173--182, 1951.

\bibitem[FH03]{Fouque/Han2003}
Jean-Pierre Fouque and Chuan-Hsiang Han.
\newblock Pricing {A}sian options with stochastic volatility.
\newblock {\em Quantitative Finance}, 3(5):353--362, 2003.

\bibitem[FM08]{Fusai/Meucci2008}
Gianluca Fusai and Attilio Meucci.
\newblock Pricing discretely monitored {Asian} options under {L\'evy}
  processes.
\newblock {\em Journal of Banking and Finance}, 32(10):2076--2088, 2008.

\bibitem[Gla04]{Glasserman}
Paul Glasserman.
\newblock {\em Monte Carlo Methods in Financial Engineering}.
\newblock Springer, 2004.

\bibitem[GM11]{Gobet/Miri2011}
Emmanuel Gobet and Mohammed Miri.
\newblock Weak approximation of averaged diffusion processes.
\newblock Preprint, 2011.

\bibitem[GY93]{Geman/Yor1993}
H\'elyette Geman and Marc Yor.
\newblock Bessel processes, {Asian} options, and perpetuities.
\newblock {\em Mathematical Finance}, 3(4):349--375, 1993.

\bibitem[Hes93]{Heston}
Steven~L. Heston.
\newblock A closed-form solution for options with stochastic volatility with
  applications to bond and currency options.
\newblock {\em Rev. Fin. Studies}, 6:327--343, 1993.

\bibitem[HKK06]{HKK2006}
Friedrich Hubalek, Jan Kallsen, and Leszek Krawczyk.
\newblock Variance-optimal hedging for processes with stationary independent
  increments.
\newblock {\em The Annals of Applied Probability}, 16(2):853--885, 2006.

\bibitem[HS09]{HubalekSgarra2009}
Friedrich Hubalek and Carlo Sgarra.
\newblock On the {E}sscher transforms and other equivalent martingale measures
  for {B}arndorff-{N}ielsen and {S}hephard stochastic volatility models with
  jumps.
\newblock {\em Stochastic Processes and their Applications}, 119(7):2137--2157,
  2009.

\bibitem[HS11]{HS2011}
Friedrich Hubalek and Carlo Sgarra.
\newblock On the explicit evaluation of the geometric {A}sian options in
  stochastic volatility models with jumps.
\newblock {\em Journal of Computational and Applied Mathematics},
  235(11):3355--3365, 2011.

\bibitem[Kal06]{Kallsen2006ASV}
Jan Kallsen.
\newblock A didactic note on affine stochastic volatility models.
\newblock In {\em From stochastic calculus to mathematical finance}, pages
  343--368. Springer, Berlin, 2006.

\bibitem[Kat01]{Kat}
Harry Kat.
\newblock {\em Structured Equity Derivatives: The Definite Guide to Exotic
  Options and Structured Notes}.
\newblock John Wiley, 2001.

\bibitem[KMKV11]{KMKV2011}
Jan Kallsen, Johannes Muhle-Karbe, and Moritz Vo{\ss}.
\newblock Pricing options on variance in affine stochastic volatility models.
\newblock {\em Mathematical Finance}, 21(4):627--641, 2011.

\bibitem[KR08]{KR2008}
Martin Keller-Ressel.
\newblock {\em Affine processes --- Theory and applications in finance}.
\newblock Dissertation, Vienna University of Technology, 2008.

\bibitem[KR11]{KellerRessel2010}
Martin Keller-Ressel.
\newblock Moment explosions and long-term behavior of affine stochastic
  volatility models.
\newblock {\em Mathematical Finance}, 21(1):73--98, 2011.

\bibitem[KRST10]{KellerSchachTeich2010}
Martin Keller-Ressel, Walter Schachermayer, and Joseph Teichmann.
\newblock Affine processes are regular.
\newblock Preprint, 2010.

\bibitem[KW11]{KW2011}
Bara Kim and In-Suk Wee.
\newblock Pricing of geometric {Asian} options under {Heston}'s stochastic
  volatility model.
\newblock {\em Quantitative Finance}, iFirst:1--15, 2011.

\bibitem[Mer76]{Merton1976}
R.~Merton.
\newblock Option pricing when underlying stock returns are discontinuous.
\newblock {\em The Journal of Financial Economics}, 3:125--144, 1976.
\newblock check!

\bibitem[MP98]{Milevsky1998}
M.A. Milevsky and S.E. Posner.
\newblock {A}sian options, the sum of lognormals and the reciprocal gamma
  distribution.
\newblock {\em Journal of Financial and Quantitative Analysis}, 33(3):409--422,
  1998.

\bibitem[NV03]{NicolatoVenardos2003}
Elisa Nicolato and Emmanouil Venardos.
\newblock Option pricing in stochastic volatility models of the
  {O}rnstein-{U}hlenbeck type.
\newblock {\em Mathematical Finance}, 13(4):445--466, 2003.

\bibitem[Pap07]{Pap2007}
Antonis Papapantoleon.
\newblock {\em Applications of semimartingales and {L\'evy} processes in
  finance: duality and valuation}.
\newblock Phd thesis, University of Freiburg, 2007.

\bibitem[Pen06]{Pen2006}
Bin Peng.
\newblock Pricing geometric {Asian} options under the {CEV} process.
\newblock {\em International Economic Journal}, 20(4):515--522, 2006.

\bibitem[PZ03]{PZ2003}
Andrei~D. Polyanin and Valentin~F. Zaitsev.
\newblock {\em Handbook of exact solutions for ordinary differential
  equations}.
\newblock Chapman \& Hall/CRC, Boca Raton, FL, second edition, 2003.

\bibitem[Rei72]{Reid1972}
William~T. Reid.
\newblock {\em Riccati differential equations}.
\newblock Academic Press, New York, 1972.
\newblock Mathematics in Science and Engineering, Vol. 86.

\bibitem[Ron95]{Ron1995}
A.~Ronveaux, editor.
\newblock {\em Heun's differential equations}.
\newblock Oxford Science Publications. The Clarendon Press Oxford University
  Press, New York, 1995.

\bibitem[Sch08]{Schroeder2008}
Michael Schr{\"o}der.
\newblock On constructive complex analysis in finance: explicit formulas for
  {A}sian options.
\newblock {\em Quarterly of Applied Mathematics}, 66(4):633--658, 2008.

\bibitem[Sep08]{Sep2008}
Artur Sepp.
\newblock Pricing options on realized variance in the {Heston} model with jumps
  in returns and volatility.
\newblock {\em Journal of Computational Finance}, 11(4):33--70, 2008.

\bibitem[SGD00]{Goovaerts2000}
S.~Simon, M.J. Goovaerts, and J.~Dhaene.
\newblock An easily computable upper bound for the price of an arithmetic
  {A}sian option.
\newblock {\em Insurance: Mathematics and Economics}, 26(2-3):175--183, 2000.

\bibitem[SK10]{SK2010}
B.~D. Sleeman and V.~B. Kuznetsov.
\newblock Heun functions.
\newblock In {\em N{IST} handbook of mathematical functions}, pages 709--721.
  U.S. Dept. Commerce, Washington, DC, 2010.

\bibitem[Sla60]{Sla1960}
L.~J. Slater.
\newblock {\em Confluent hypergeometric functions}.
\newblock Cambridge University Press, New York, 1960.

\bibitem[VDL{\etalchar{+}}06]{Goovaerts2006}
M.~Vanmaele, G.~Deelstra, J.~Liinev, J.~Dhaene, and M.~J. Goovaerts.
\newblock Bounds for the price of discrete arithmetic {A}sian options.
\newblock {\em Journal of Computational and Applied Mathematics},
  185(1):51--90, 2006.

\bibitem[VX04]{Vecer/Xu2004}
Jan Ve{\v{c}}e{\v{r}} and Mingxin Xu.
\newblock Pricing {A}sian options in a semimartingale model.
\newblock {\em Quantitative Finance}, 4(2):170--175, 2004.

\bibitem[WHD95]{Wilmott/Dewynne/Howison}
Paul Wilmott, Sam Howison, and Jeff Dewynne.
\newblock {\em The mathematics of financial derivatives}.
\newblock Cambridge University Press, Cambridge, 1995.

\bibitem[ZO13]{ZO2013}
B.~Zhang and C.~W. Oosterlee.
\newblock Efficient pricing of {E}uropean-style {A}sian options under
  exponential {L}\'evy processes based on {F}ourier cosine expansions.
\newblock {\em SIAM Journal on Financial Mathematics}, 4(1):399--426, 2013.

\end{thebibliography}

%\newpage
%\tableofcontents
%

\end{document}